\documentclass[conference,a4paper,onecolumn]{IEEEtran}

\usepackage[cmex10]{amsmath}
\usepackage{graphicx}
\usepackage{latexsym}
\usepackage{amssymb}
\usepackage{bm}
\usepackage{psfrag}

\newtheorem{theorem}{Theorem}
\newtheorem{lemma}{Lemma}
\newtheorem{corollary}{Corollary}
\newtheorem{remark}{Remark}

\newcommand{\bs}{\boldsymbol}
\newcommand{\coef}{\mathrm{coef}}
\newcommand{\diff}[1]{d#1}
\newcommand{\pderi}[2]{\frac{\partial #1}{\partial #2}}
\newcommand{\deri}[2]{\frac{\diff{#1}}{\diff{#2}}}
\newcommand{\neib}[1]{\mathcal{N}_{\mathrm{#1}}}
\newcommand{\Ens}{\mathcal{G}}

\begin{document}

\sloppy

\title{
Weight Distribution for Non-binary Cluster LDPC Code Ensemble
}

\author{
\IEEEauthorblockN{Takayuki Nozaki}
\IEEEauthorblockA{
Kanagawa University, JAPAN\\
Email: nozaki@kanagawa-u.ac.jp
}
\and
\IEEEauthorblockN{
Masaki Maehara,
Kenta Kasai,
and Kohichi Sakaniwa}
\IEEEauthorblockA{
Tokyo Institute of Technology, JAPAN\\
Email: \{maehara, kenta,  sakaniwa\}@comm.ce.titech.ac.jp}
}
\maketitle

\begin{abstract}
In this paper, we derive the average weight distributions for the irregular non-binary cluster low-density parity-check (LDPC) code ensembles.
Moreover, we give the exponential growth rate of the average weight distribution in the limit of large code length.
We show that there exist ($2,d_{\mathrm{c}}$)-regular non-binary cluster LDPC code ensembles whose normalized typical minimum distances are strictly positive.
\end{abstract}

\section{Introduction}
Gallager invented low-density parity-check (LDPC) codes \cite{Gallager_LDPC}.
Due to the sparseness of the parity check matrices,
LDPC codes are efficiently decoded by the belief propagation (BP) decoder. 
Optimized LDPC codes can exhibit performance very close 
to the Shannon limit \cite{richardson01design}.
Davey and MacKay \cite{DaveyMacKayGFq} 
have found that non-binary LDPC codes can outperform binary ones.

The LDPC codes are defined by sparse parity check matrices or sparse Tanner graphs.
For the non-binary LDPC codes, the Tanner graphs are represented by bipartite graph with variable nodes and check nodes and labeled edges.
The LDPC codes defined by Tanner graphs with the variable nodes of degree $d_{\mathrm v}$
and the check nodes of degree $d_{\mathrm{c}}$ 
are called $(d_{\mathrm{v}}, d_{\mathrm{c}})$-regular LDPC codes. 
It is empirically known that $(2,d_{\mathrm c})$-regular non-binary LDPC codes
exhibit good decoding performance among other LDPC codes for the non-binary LDPC code defined over Galois field of order greater than $32$ \cite{Hu03regularand}. 

Savin and Declercq proposed the non-binary cluster LDPC codes \cite{6034183}.
For the non-binary cluster LDPC code, each edge in the Tanner graphs is labeled by {\it cluster} which is a full-rank $p \times r$ binary matrix, where $p \ge r$.
In \cite{6034183}, Savin and Declercq showed that there exist $(2,d_{\mathrm{c}})$-regular non-binary cluster LDPC ensembles whose minimum distance grows linearly with the code length.

Deriving the weight distribution is important to analyze the decoding performances for the linear codes.
In particular, in the case for LDPC codes, weight distribution gives
a bound of decoding error probability under maximum likelihood decoding \cite{959254}
and error floors under belief propagation decoding and maximum likelihood decoding \cite{di_wd} \cite{mct}.

Studies on weight distribution for non-binary LDPC codes date back to \cite{Gallager_LDPC}.
Gallager derived the symbol-weight distribution of Gallager code ensemble defined over $\mathbb{Z}/q\mathbb{Z}$ \cite{Gallager_LDPC}.
Kasai et al.\ derived the average symbol and bit weight distributions and the exponential growth rates for the irregular non-binary LDPC code ensembles defined over Galois field $\mathbb{F}_q$, and showed that the normalized typical minimum distance does not monotonically grow with $q$ \cite{Kasai_nbw}.
Andriyanova et al.\ derive the bit weight distributions and the exponential growth rates for the regular non-binary LDPC code ensembles defined over Galois field and general linear groups \cite{5205662}.

In this paper, we derive the average symbol and bit weight distributions for the irregular non-binary cluster LDPC code ensembles.
Moreover, we give the exponential growth rate of the average weight distributions in the limit of large code length.

The remainder of this paper is organized as follows:
Section \ref{sec:pre} defines the irregular non-binary cluster LDPC code ensemble.
Section \ref{sec:wd} derives the average weight distributions for the irregular non-binary LDPC code ensembles.
Section \ref{sec:gr} gives the exponential growth rate of the average weight distributions in the limit of large code length and shows some numerical examples for the exponential growth rate.

\section{Preliminaries \label{sec:pre}}
In this section, we review non-binary cluster LDPC code \cite{6034183} and define the irregular non-binary cluster LDPC code ensemble.
We introduce some notations used throughout this paper.

\subsection{Non-binary Cluster LDPC Code}
The LDPC codes are defined by sparse parity check matrices or sparse Tanner graphs.
For the non-binary LDPC codes, the Tanner graphs are represented by bipartite graphs with variable nodes and check nodes and {\it labeled} edges.

For the non-binary cluster LDPC codes, each edge in the Tanner graphs is labeled by {\it cluster} which is a full-rank $p \times r$ binary matrix, where $p \ge r$.
Let $\mathbb{F}_2$ be the finite field of order 2.
Note that the non-binary LDPC codes defined by Tanner graphs labeled by general linear group $\mathrm{GL}(p,\mathbb{F}_2)$ are special cases for the non-binary cluster LDPC code with $p=r$.

We denote the cluster in the edge between the $i$-th variable node and the $j$-th check node, by $h_{j,i}$.
For the cluster LDPC codes, $r$-bits are assigned to each variable node in the Tanner graphs.
We refer to the $r$-bits assigned to the $i$-th variable node as {\it symbol} assigned to the $i$-th variable node, and denote it by ${\bs x}_{i}\in\mathbb{F}_2^r$.

For integers $a,b$, we denote the set of integers between $a$ and $b$, as $[a;b]$.
More precisely, we define
\begin{align*}
 [a;b] := 
 \begin{cases}
  \{n\in\mathbb{N}\mid a\le n \le b\}, & a\le b,\\
  \emptyset = \{\}, & a>b.
 \end{cases}
\end{align*}
The non-binary cluster LDPC code defined by a Tanner graph $\mathtt{G}$ is given as follows:
\begin{align*}
 C(\mathtt{G})
  =
 \Bigl\{( {\bs x}_1,\dots,{\bs x}_N)\in (\mathbb{F}_2^r)^N 
  \mid
 {\textstyle\sum_{i\in \neib{c}(j)}}h_{j,i}{\bs x}_i^T ={\bs 0}^T\in\mathbb{F}_2^p
 ~~\forall j\in[1;M]\Bigr\},
\end{align*}
where $\neib{c}(j)$ represents the set of indexes of the variable nodes adjacent to the $j$-th check node.
Note that $N$ is called symbol code length and the bit code length $n$ is given by $r N$.

\subsection{Irregular Non-binary Cluster LDPC Code Ensemble}
Let $\mathcal{L}$ and $\mathcal{R}$ be the sets of degrees of
the variable nodes and the check nodes, respectively.
Irregular non-binary cluster LDPC codes are characterized with
the number of variable nodes $N$, the size of cluster $p,r$
and a pair of {\it degree distribution},
$\lambda(x) = \sum_{i\in{\cal L}}\lambda_{i}x^{i-1}$ and
$\rho(x) = \sum_{i\in{\cal R}}\rho_{i}x^{i-1}$,
where $\lambda_i$ and $\rho_{i}$ are the fractions of the edges 
connected to the variable nodes and the check nodes of degree $i$,
respectively.

The total number of the edges in the Tanner graph is
\begin{equation*}
 E
  :=
 \frac{N}{\int_{0}^{1}\lambda(x)\diff{x}.}
\end{equation*}
The number of check node $M$ is given by 
\begin{equation*}
 M = 
 \biggl(
   \frac{\int_0^1 \rho(x)\diff{x}}{\int_0^1 \lambda(x)\diff{x}}
 \biggr) N
  =: \kappa N.
\end{equation*}
Let $L_i$ and $R_j$ be the fraction of the variable nodes of degree $i$
and the check nodes of degree $j$, respectively, i.e.,
\begin{align*}
 L_i := 
 \frac{\lambda_{i}}{i\int_{0}^{1}\lambda(x)\diff{x}},
\quad
 R_j := 
 \frac{\rho_{j}}{ j\int_{0}^{1}\rho(x)   \diff{x}}.
\end{align*}
The design rate is given as follows:
\begin{equation*}
 1 - \frac{\kappa p}{r}.
\end{equation*}

Assume that we are given the number of variable nodes $N$, 
the size of cluster $p,r$ and the degree distribution pair $(\lambda,\rho)$.
An irregular non-binary cluster LDPC code ensemble $\Ens(N,p,r,\lambda,\rho)$ 
is defined as the following way.
There exist $L_i N$ variable nodes of degree $i$ and 
$R_j M$ check nodes of degree $j$.
A node of degree $i$ has $i$ sockets for its connected edges.
Consider a permutation $\pi$ on the number of edges.
Join the $i$-th socket on the variable node side to the $\pi(i)$-th socket
on the check node side.
The bipartite graphs are chosen with equal probability
from all the permutations on the number of edges.
Each cluster in an edge is chosen a full-rank $p\times r$ binary matrix 
with equal probability.

\section{Weight Distribution for Non-binary Cluster LDPC Code \label{sec:wd}}
In this section, we derive the average symbol and bit weight distribution for the irregular non-binary cluster LDPC code ensemble $\Ens(N,p,r,\lambda,\rho)$.

We denote the $r$-bit representation of ${\bs x}_i\in\mathbb{F}_{2}^r$, by $(x_{i,1},\dots,x_{i,r})$.
For a given codeword ${\bs x} = ({\bs x}_1,{\bs x}_2,\dots,{\bs x}_N)$,
we denote the symbol and bit weight of ${\bs x}$, by $w({\bs x})$ and $w_{\mathrm{b}}({\bs x})$.
More precisely, we define
\begin{align*}
 w({\bs x}) 
  &:=
 |\{ i\in[1;N] \mid {\bs x}_i \neq {\bs 0}\}|, \\
 w_{\mathrm{b}}({\bs x})
  &:=
 |\{ (i,j)\in[1;N]\times[1;r] \mid x_{i,j} \neq 0\}|.
\end{align*}
For a given Tanner graph $\mathtt{G}$, let $A^{\mathtt{G}}(\ell)$ (resp.\ $A^{\mathtt{G}}_{\mathrm{b}}(\ell)$) be the number of codeword of symbol and (resp.\ bit) weight $\ell$ in $C(\mathtt{G})$, i.e.,
\begin{align*}
 &A^{\mathtt{G}}(\ell)
  =
 |\{ {\bs x}\in C(\mathtt{G}) 
   \mid  
   w({\bs x}) = \ell
  \}|,  \\
 &A^{\mathtt{G}}_{\mathrm{b}}(\ell)
  =
 |\{ {\bs x}\in C(\mathtt{G}) 
   \mid  
   w_{\mathrm{b}}({\bs x}) = \ell
  \}|.
\end{align*}
For the irregular non-binary cluster LDPC code ensemble $\Ens(N,r,p,\lambda,\rho)$, we denote the average number of codewords of symbol and bit weight $\ell$, by $A(\ell)$ and $A_{\mathrm{b}}(\ell)$, respectively.
Since each Tanner graph in the ensemble $\mathcal{G}=\Ens(N,r,p,\lambda,\rho)$ is chosen with uniform probability, the following equations hold:
\begin{align*}
 A(\ell) 
  = 
 \frac{1}{|\mathcal{G}|}\sum_{\mathtt{G}\in\mathcal{G}}  A^{\mathtt{G}}(\ell),
 \quad
 A_{\mathrm{b}}(\ell) 
  = 
 \frac{1}{|\mathcal{G}|}\sum_{\mathtt{G}\in\mathcal{G}} 
  A^{\mathtt{G}}_{\mathrm{b}}(\ell).
\end{align*}

Since the number of full-rank binary $p\times r$ matrix is $\prod_{i=0}^{r-1}(2^p-2^i)$, 
the number of codes in the ensemble $\mathcal{G} = \Ens(N,r,p,\lambda,\rho)$ is given as
\begin{align}
|\mathcal{G}| 
  = 
 E ! \Biggl\{ \prod_{i=0}^{r-1}  (2^p-2^i) \Biggr\}^{E}.
 \label{eq:num_code}
\end{align}

\subsection{Symbol Codeword Weight Distribution}
First, we will derive the average symbol weight distributions for the irregular non-binary cluster LDPC code ensembles.
\begin{theorem} \label{the:1}
The average number $A(\ell)$ of codewords of symbol weight $\ell$ for the irregular non-binary cluster LDPC code ensemble $\Ens(N,p,r,\lambda,\rho)$ is 
\begin{align}
&A(\ell)
  =
 \sum_{k=0}^{E}
 \frac{(2^r-1)^{\ell}\coef\bigl((P(s,t)Q(u))^N,s^{\ell}t^{k}u^k\bigr)}
 {\binom{E}{k}(2^p-1)^k}, \label{eq:Aell} \\
&P(s,t)
 :=
 \prod_{i\in\mathcal{L}} \bigl(1+st^i\bigr)^{L_i},\quad
Q(u)
 :=
 \prod_{j\in\mathcal{R}} f_j(u)^{\kappa R_j}, \notag \\
&f_j(u)
 :=
 \frac{1}{2^p}\bigl[\{1+(2^p-1)u\}^j+(2^p-1)(1-u)^j\bigr], \label{eq:fj}
\end{align}
where
$\coef(g(s,t,u),s^i t^j u^k)$ is the coefficient of the term $s^i t^j u^k$
of the polynomial $g(s,t,u)$.
\end{theorem}
\begin{proof}
We follow the similar way in \cite[Theorem 1]{Kasai_nbw}.

We refer to an edge as {\it active} if the edge connects to a variable node to which is assigned a non-zero symbol.
We will count the average number of codewords $A(\ell,k)$ with symbol weight $\ell$ and the number of active edges $k$.

Firstly, we count the edge constellations satisfying the constraints of the variable nodes.
Consider a variable node $\mathtt{v}$ of degree $i$.
Define the parameter $\tilde{\ell}$ as 1 if a non-zero symbol is assigned to the variable node $\mathtt{v}$, and otherwise 0.
For a given $\tilde{\ell}\in[0,1]$ and $\tilde{k}\in[0,i]$, let $a_i(\tilde{\ell},\tilde{k})$ be the number of constellations of $\tilde{k}$ active edges which stem from a variable node of degree $i$.
The $i$ edges connected to $\mathtt{v}$ are active
if and only if a non-zero symbol is assigned to the variable node $\mathtt{v}$.
Hence, we have
\begin{align*}
 a_i(\tilde{\ell},\tilde{k})
  =
 \begin{cases}
  1, & \tilde{\ell} = 0,~\tilde{k}=0, \\
  2^r-1, & \tilde{\ell} = 1,~\tilde{k}=i, \\
  0, & \text{otherwise}.
 \end{cases}
\end{align*}
The generating function of $a_i(\tilde{\ell},\tilde{k})$ is written as follows:
\begin{align*}
 \sum_{\tilde{\ell},\tilde{k}} a_i(\tilde{\ell},\tilde{k}) s^{\tilde{\ell}}t^{\tilde{k}}
  =
 1+(2^r-1)st^{i}.
\end{align*}
Since there are $L_i N$ variable nodes of degree $i$, for a given $\ell$ and $k$,
the number of edge constellations satisfying constraints of the $N$ variable nodes in the Tanner graph is given by 
\begin{align*}
 \coef
 \Bigl(
  {\textstyle\prod_{i\in\mathcal{L}}} \{ 1+(2^r-1)st^i \}^{L_i N}, 
  s^{\ell} t^{k}
 \Bigr).
\end{align*}
This equation is simplified as follows:
\begin{align}
 (2^r-1)^{\ell}
 \coef
 \Bigl(
  {\textstyle\prod_{i\in\mathcal{L}}} (1+st^i)^{L_i N}, s^{\ell} t^{k}
 \Bigr).
 \label{eq:v_const}
\end{align}

Secondly, we count the edge constellations satisfying all the constraints of the check nodes.
Consider a check node $\mathtt{c}$ of degree $j$.
Let $m_j(\tilde{k})$ be the number of constellations of the $\tilde{k}$ active edges satisfying a check node of degree $j$.
In other words, 
\begin{align*}
 m_j(\tilde{k})
  =
 |\{ ({\bs y}_1,{\bs y}_2,\dots,{\bs y}_j)\in(\mathbb{F}_{2}^p)^j
   \mid 
  {\textstyle \sum_{i=1}^j}{\bs y}_i = {\bs 0},
  |\{i\mid {\bs y}_i \neq {\bs 0}\}| = \tilde{k} \}|.
\end{align*}
As in \cite[Eq.~(5.3)]{Gallager_LDPC}, $m_j(\tilde{k})$ is given as follows:
\begin{align*}
 m_j(\tilde{k})
  =
 \binom{j}{\tilde{k}}\frac{1}{2^p}
 \Bigl\{ (2^p-1)^{\tilde{k}} + (-1)^{\tilde{k}}(2^p-1) \Bigr\}
\end{align*}
The generating function of $m_j(\tilde{k})$ is written as follows:
\begin{align*}
 f_j(u)
  &=
 \sum_{\tilde{k}} m_j(\tilde{k}) u^{\tilde{k}} \\
  &=
 \frac{1}{2^p} \Bigl[ \{1+(2^p-1)u\}^j + (2^p-1)(1-u)^j \Bigr].
\end{align*}
Since there are $\kappa R_j N$ check nodes of degree $j$, 
for a given number of active edge $k$,
the number of the constellations satisfying all the constraints of the check nodes is given as:
\begin{align}
 \coef \Bigl( {\textstyle \prod_{j\in\mathcal{R}}} f_j(u)^{\kappa R_j N}, 
 u^k \Bigr).
 \label{eq:c_const}
\end{align}

Thirdly, we count the edge permutation and the number of clusters which satisfy the edge constraints.
For a given number of active edge $k$, the number of permutations of edges is given by $k!(E-k)!$ and the number of clusters which satisfy the edge constraints
is equal to 
$$\left\{\prod_{i=1}^{r-1}(2^p-2^i)\right\}^{k}
 \left\{\prod_{i=0}^{r-1}(2^p-2^i)\right\}^{E-k}.$$
Hence, for a given number of active edge $k$, the number of choices for the permutation of edges and clusters is 
\begin{align}
 k!(E-k)!
 \left\{ \prod_{i=1}^{r-1} (2^p-2^i)\right\}^{k}
 \left\{ \prod_{i=0}^{r-1} (2^p-2^i)\right\}^{E-k}.
 \label{eq:e_const}
\end{align}

By multiplying Eqs.~\eqref{eq:v_const}, \eqref{eq:c_const} and \eqref{eq:e_const}, and dividing by Eq.~\eqref{eq:num_code}, we obtain the number of codewords $A(\ell,k)$ with symbol weight $\ell$ and the number of active edges $k$ as 
\begin{align*}
 A(\ell,k)
  =
 \frac{(2^r-1)^{\ell}\coef\bigl((P(s,t)Q(u))^N,s^{\ell}t^{k}u^k\bigr)}
 {\binom{E}{k}(2^p-1)^k}.
\end{align*}
Since $A(\ell) = \sum_{k=0}^{E} A(\ell, k)$, we get Theorem \ref{the:1}.
\end{proof}

Theorem \ref{the:1} gives the following corollary.
\begin{corollary} \label{cor:1}
For the irregular non-binary cluster LDPC code ensemble 
$\Ens(N,p,r,\lambda,\rho)$,
the following equations hold:
\begin{align*}
 &A(0) = 1, \\
 &A(N) = 
  \frac{(2^r-1)^N
\prod_{j\in\mathcal{R}} \bigl\{ (2^p-1)^{j} + (-1)^{j}(2^p-1) \bigr\}^{\kappa R_j N}}
{(2^p-1)^E (2^p)^{\kappa N}}.
\end{align*}
\end{corollary}

\subsection{Bit Codeword Weight Distribution}
In a similar way to the average symbol weight distribution, we are able to derive the average bit weight distribution for the irregular non-binary cluster LDPC code ensemble $\Ens(N,r,p,\lambda,\rho)$.
At first, we consider a variable node of degree $i$.
For a given bit weight $\tilde{\ell}\in[0,r]$, let $a_{\mathrm{b},i}(\tilde{\ell},\tilde{k})$ be the number of constellations of $\tilde{k}$ active edges which stem from a variable node of degree $i$.
From the definition of active edges, we have
\begin{align*}
 a_{\mathrm{b},i}(\tilde{\ell},\tilde{k})
  =
 \begin{cases}
  1, & \tilde{\ell}=0, \tilde{k} = 0, \\
  \binom{r}{\tilde{\ell}}, & \tilde{\ell} \in[1;r], \tilde{k} = i, \\
  0, & \text{otherwise}.
 \end{cases}
\end{align*}
The generating function of $a_{\mathrm{b},i}(\tilde{\ell},\tilde{k})$ is given as:
\begin{align*}
 \sum_{\tilde{\ell},\tilde{k}} a_{\mathrm{b},i}(\tilde{\ell},\tilde{k}) 
  s^{\tilde{\ell}}t^{\tilde{k}}
  =
 1+\{ (1+s)^r-1 \} t^i.
\end{align*}
Since there are $L_i N$ variable nodes of degree $i$,
the number of constellations of $k$ active edges satisfying constraints of the $N$ variable nodes with bit weight $\ell$ is
\begin{align*}
 \coef
  \Bigl(
  {\textstyle\prod_{i\in\mathcal{L}}}  [1+\{ (1+s)^r-1 \} t^i]^{L_i N}, 
  s^{\ell}t^{k}
 \Bigr).
\end{align*}
By using this equation, in a similar way to proof of the average symbol weight distributions, we obtain the average number $A_{\mathrm{b}}(\ell)$ of codewords of bit weight $\ell$ as follows:
\begin{theorem} \label{the:2}
 Let $n = rN $ be the bit code length.
 Define $f_{j}(u)$ as in Eq.~\eqref{eq:fj}.
 The average number $A_{\mathrm{b}}(\ell)$ of codewords of bit weight $\ell$ for the irregular non-binary cluster LDPC code ensemble $\Ens(N,p,r,\lambda,\rho)$ is 
\begin{align*}
&A_{\mathrm{b}}(\ell)
  =
 \sum_{k=0}^{E}
 \frac{\coef\bigl((P_{\mathrm{b}}(s,t)Q_{\mathrm{b}}(u))^n,s^{\ell}t^{k}u^k\bigr)}
 {\binom{E}{k}(2^p-1)^k}, \\
&P_{\mathrm{b}}(s,t)
 :=
 \prod_{i\in\mathcal{L}}  [1+\{ (1+s)^r-1 \} t^i]^{L_i / r}
,\\
&Q_{\mathrm{b}}(u)
 :=
 \prod_{j\in\mathcal{R}} f_j(u)^{\kappa R_j / r}.
\end{align*}
\end{theorem}

Theorem \ref{the:2} gives the following corollary.
\begin{corollary} \label{cor:2}
 For the irregular non-binary cluster LDPC code ensemble 
 $\Ens(N,p,r,\lambda,\rho)$,
the following equations hold:
\begin{align*}
 &A_{\mathrm{b}}(0) = 1, \\
 &A_{\mathrm{b}}(n) 
= 
\frac{\prod_{j\in\mathcal{R}}
\bigl\{ (2^p-1)^j+(-1)^j(2^p-1)\bigr\}^{\kappa R_j N}}{(2^p-1)^E (2^p)^{\kappa N}}
.
\end{align*}
\end{corollary}

\section{Asymptotic Analysis \label{sec:gr}}
In this section, we investigate the asymptotic behavior of the average symbol and bit weight distributions for the non-binary cluster LDPC code ensembles in the limit of large code length.

\subsection{Growth rate}
We define
\begin{align*}
 &\gamma(\omega) 
  :=
 \lim_{N\to \infty} \frac{1}{N} \log_{2^r} A(\omega N)
  =
 \lim_{N\to \infty} \frac{1}{r N} \log_{2} A(\omega N), \\
 &\gamma_{\mathrm{b}}(\omega_{\mathrm{b}} )
  :=
 \lim_{n\to \infty} \frac{1}{n} \log_2 A_{\mathrm{b}}(\omega_{\mathrm{b}} n),
\end{align*}
and refer to them as the {\it exponential growth rate} or simply {\it growth rate} of the average number of codewords in terms of symbol and bit weight, respectively.
To simplify the notation, we denote $\log_2(\cdot)$ as $\log(\cdot)$.

With the growth rate, we can roughly estimate the average number of codewords of symbol weight $\omega N$ (resp.\ bit weight $\omega_{\mathrm{b}} n$) by
\begin{align*}
 A(\omega N) \sim (2^r)^{\gamma(\omega)N}, \quad
 (\text{resp.}~~
 A_{\mathrm{b}}(\omega_{\mathrm{b}} n) \sim 2^{\gamma_{\mathrm{b}}(\omega_{\mathrm{b}})n},
 )
\end{align*}
where $a_N\sim b_N$ means that $\lim_{N\to\infty}(1/N)\log a_N/ b_N = 0$.

\subsubsection{Growth Rate of Symbol Weight Distribution}
Since the number of terms in Eq.~\eqref{eq:Aell} is equal to $E+1$,
we get
\begin{align*}
 \max_{k\in[0;E]} A(\ell,k)
  \le
 A(\ell)
  \le
 (E+1) \max_{k\in[0;E]} A(\ell,k).
\end{align*}
Therefore, we have
\begin{align*}
 \lim_{N\to\infty} \frac{1}{r N}\log A(\ell)
  =
 \lim_{N\to\infty} \frac{1}{r N}\max_{k\in[0;E]} \log A(\ell,k)
\end{align*}
To calculate this equation, we introduce the following lemma.
\begin{lemma}{{\rm \cite[Theorem 2]{1302293} }} \label{lem:1}
Let $\gamma>0$ be some rational number and
let $p(x_1,x_2,\dots,x_m)$ be a function such that 
$p(x_1,x_2,\dots,x_m)^{\gamma}$ is a multivariate polynomial 
with non-negative coefficients.
Let $\alpha_k >0$ be some rational numbers for $k\in[1;m]$ and
let $n_i$ be the series of all indexes $j$ such that $j/\gamma$ is an
integer and 
$\coef(p(x_1,\dots,x_m)^{j},x_1^{\alpha_1 j} \cdots x_m^{\alpha_m j})\neq 0$.
Then 
\begin{align}
 \lim_{i\to\infty}&\frac{1}{n_i}\log 
  \coef(p(x_1,\dots,x_m)^{n_i}, 
        (x_1^{\alpha_1}\cdots x_m^{\alpha_m})^{n_i})
 =
  \inf_{x_1,\dots,x_m > 0} \log 
   \frac{p(x_1,\dots,x_m)}{x_1^{\alpha_1}\cdots x_m^{\alpha_m}}.
 \notag
\end{align}
A point $(x_1,\dots,x_m)$ achieves the minimum of the function
$
 p(x_1,\dots,x_m)/(x_1^{\alpha_1}\dots x_m^{\alpha_m}),
$
if and only if it satisfies the following equation for all $k\in[1;m]$:
\begin{align}
 x_k\pderi{p(x_1,\dots,x_m)^{\gamma}}{x_k} 
 -\gamma\alpha_k p(x_1,\dots,x_m)^{\gamma}
  =
 0. \notag
\end{align}
\end{lemma}

From Theorem \ref{the:1} and Lemma \ref{lem:1}, we obtain the following theorem.
\begin{theorem} \label{the:3}
 Define $\omega = \ell/N$, $\beta := k/N$ and $\epsilon := E/N$.
 The growth rate $\gamma(\omega)$ of the average number of codewords of normalized symbol weight $\omega$ for the irregular non-binary cluster LDPC code ensemble $\Ens(N,p,r,\lambda,\rho)$ with sufficiently large $N$
is given by, for $0<\omega<1$,
\begin{align} 
\gamma(\omega)
 &=
 \sup_{\beta>0} \inf_{s>0,t>0,u>0}
 \frac{1}{r}  \biggl[
   \log P(s,t)
  +\log Q(u) 
  -\epsilon h \biggl(\frac{\beta}{\epsilon}\biggr)
 -\beta\log (t u(2^p-1))
 -\omega\log \biggl(\frac{s}{2^r-1}\biggr)
 \biggr] \notag \\
&=: 
 \sup_{\beta>0} \inf_{s>0,t>0,u>0} \gamma(\omega,\beta,s,t,u) \notag \\
&=:
 \sup_{\beta>0}\gamma(\omega,\beta), \label{eq:gr_S}
\end{align}
where $h(x) := - x \log x - (1-x) \log (1-x)$ for $0<x<1$.
 A point $(s,t,u)$  which achieves the minimum of the function
 $\gamma(\omega, \beta, s, t, u)$
 is given in a solution of the following equations:
\begin{align}
\omega
 &=
 \frac{s}{P}\pderi{P}{s}
  =
 \sum_{i\in\mathcal{L}} L_i \frac{s t^i}{1+s t^i}, \label{eq:Ps}\\
\beta
 &=
 \frac{t}{P}\pderi{P}{t}
  =
 \sum_{i\in\mathcal{L}} L_i \frac{i s t^i}{1+s t^i}, \label{eq:Pt} \\
\beta
 &=
  \frac{u}{Q}\pderi{Q}{u}
 =
 \sum_{j\in\mathcal{R}} \kappa R_j
 \frac{u}{f_j(u)}\pderi{f_j}{u}(u),
  \label{eq:Qu}
\end{align}
where
\begin{align*}
\pderi{f_j}{u}(u)
 =
j\frac{2^p-1}{2^p} [\{1+(2^p-1)u\}^{j-1}-(1-u)^{j-1}].
\end{align*}
The point $\beta$ which gives the maximum of $\gamma(\omega,\beta)$ needs to satisfy the stationary condition
\begin{align}
 \beta
  =
 (2^p-1)t u (\epsilon-\beta). \label{eq:st_S}
\end{align}
\end{theorem}

From Corollary \ref{cor:1} and the definition of growth rate,
we derive the growth rate of average number of codewords with $\omega = 0, 1$ as follows:
\begin{corollary} \label{cor:3}
 For the irregular non-binary cluster LDPC code ensemble $\Ens(N, p,r,\lambda,\rho)$ in the limit of large symbol code length $N$, the following equations hold:
\begin{align*}
 &\gamma(0) 
  = 
 0, \\
 &\gamma(1) 
  = 
 \frac{1}{r} \Biggl[
   \log (2^r-1)
  -\epsilon \log(2^p-1)
  -\kappa p 
  +\sum_{j\in\mathcal{R}}
   \kappa R_j \log \{ (2^p-1)^j + (-1)^j (2^p-1) \}
 \Biggr].
\end{align*}
Moreover, by letting $p,r$ tend to infinity with a fixed ratio,
we have
\begin{align*}
 \gamma(1) 
  \to 
 1-\frac{\kappa p}{r},
\end{align*}
namely, 
$\gamma(1)$ tends to the design rate.
\end{corollary}

For a fixed normalized symbol weight $\omega$, the intermediate variables $s,t,u$ and $\beta$ are derived from Eqs.~\eqref{eq:Ps}, \eqref{eq:Pt}, \eqref{eq:Qu} and \eqref{eq:st_S}.
Hence, the intermediate variables $s,t,u$ and $\beta$ are represented as functions of $\omega$. 
Thus, we denote those intermediate variables, by $s(\omega), t(\omega), u(\omega), \beta(\omega)$.

The derivation of $\gamma(\omega)$ in terms of $\omega$ is simply expressed as following lemma.
\begin{lemma} \label{lem:2}
 For $s>0$ such that Eqs.~\eqref{eq:Ps}, \eqref{eq:Pt}, \eqref{eq:Qu} and \eqref{eq:st_S} hold,
we have
 \begin{align*}
  \deri{\gamma}{\omega}(\omega)
   =
  -\frac{1}{r}\log \frac{s(\omega)}{2^r-1}.
 \end{align*}
\end{lemma}
\begin{proof}
We follow the similar way in \cite{Awano_multi}.

For a fixed $\omega$,
we denote the point achieving the maximum of $\gamma(\omega,\beta)$ by $\hat{\beta}$ and the point achieving the minimum of $\gamma(\omega,\hat{\beta},s,t,u)$ by $(\hat{s},\hat{t},\hat{u})$.
Then, $\gamma(\omega) = \gamma(\omega,\hat{\beta},\hat{s},\hat{t},\hat{u})$ holds 
and $\hat{\beta},\hat{s},\hat{t},\hat{u}$ satisfy Eqs.~\eqref{eq:Ps}, \eqref{eq:Pt}, \eqref{eq:Qu} and \eqref{eq:st_S}.
From \eqref{eq:gr_S}, we have
\begin{align}
 \deri{\gamma(\omega)}{\omega} 
  =&
 \deri{}{\omega}\gamma(\omega,\hat{\beta},\hat{s},\hat{t},\hat{u}) 
 \notag \\
 =&
 \frac{1}{r\ln 2} \Biggl[
   \frac{1}{P}\deri{P}{\omega} 
  -\frac{\omega}{      \hat{s}} \deri{\hat{s}}{\omega}
  -\frac{\hat{\beta}}{ \hat{t}} \deri{\hat{t}}{\omega} 
  +\frac{1}{Q}\deri{Q}{\omega}
  -\frac{\hat{\beta}}{ \hat{u}} \deri{\hat{u}}{\omega}
  +\deri{\hat{\beta}}{\omega} 
    \ln \frac{\epsilon-\hat{\beta}}{(2^p -1) \hat{\beta} \hat{t} \hat{u}}
  -\ln \frac{\hat{s}}{(2^r-1)}
 \Biggr].
 \label{eq:lem3-1}
\end{align}
From \eqref{eq:Ps} and \eqref{eq:Pt}, 
we have
\begin{align*}
 \frac{1}{P}\deri{P}{\omega}
  =&
  \frac{1}{P}\pderi{P}{\hat{s}} \deri{\hat{s}}{\omega} 
 +\frac{1}{P}\pderi{P}{\hat{t}} \deri{\hat{t}}{\omega}  \notag \\
  =&
  \frac{\omega}{\hat{s}} \deri{\hat{s}}{\omega} 
 +\frac{\hat{\beta}}{\hat{t}} \deri{\hat{t}}{\omega}.
\end{align*}
In other words, the sum of the first three terms of Eq.~\eqref{eq:lem3-1} is equal to 0.
Similarly, from \eqref{eq:Qu},
we have
\begin{align*}
 \frac{1}{Q}\deri{Q}{\omega}
  =&
 \frac{1}{Q}\pderi{Q}{\hat{u}} \deri{\hat{u}}{\omega} \\
  =&
 \frac{\hat{\beta}}{\hat{u}} \deri{\hat{u}}{\omega},
\end{align*}
i.e., the sum of forth and fifth terms of Eq.~\eqref{eq:lem3-1} is equal to 0.
From \eqref{eq:st_S}, we see that the sixth term of Eq.~\eqref{eq:lem3-1} is equal to 0.
This concludes the proof.
\end{proof}

\subsubsection{Growth Rate of Bit Weight Distribution}
In a similar way to symbol weight, we can derive the growth rate for the average number of codewords of bit weight.
Hence, we omit the proofs in this section.
\begin{theorem} \label{the:4}
 Define $\omega_{\mathrm{b}} = \ell/n$, $\beta_{\mathrm{b}} := k/n$ and $\epsilon_{\mathrm{b}} := E/n$.
 The growth rate $\gamma_{\mathrm{b}}(\omega_{\mathrm{b}})$ of the average number of codewords of normalized bit weight $\omega_{\mathrm{b}}$ for the irregular non-binary cluster LDPC code ensemble $\Ens(N,p,r,\lambda,\rho)$ with sufficiently large $N$
is given by, for $0<\omega_{\mathrm{b}}<1$,
\begin{align} 
\gamma_{\mathrm{b}}(\omega_{\mathrm{b}})
 &=
 \sup_{\beta_{\mathrm{b}}>0} \inf_{s>0,t>0,u>0}
 \biggl[
   \log P_{\mathrm{b}}(s,t)
  +\log Q_{\mathrm{b}}(u) 
 -\epsilon_{\mathrm{b}} 
    h \biggl(\frac{\beta_{\mathrm{b}}}{\epsilon_{\mathrm{b}}}\biggr)
 -\beta_{\mathrm{b}} \log (t u(2^p-1))
 -\omega_{\mathrm{b}}\log s
 \biggr] \notag \\
&=: 
 \sup_{\beta_{\mathrm{b}} >0} \inf_{s>0,t>0,u>0} 
 \gamma_{\mathrm{b}}(\omega_{\mathrm{b}},\beta_{\mathrm{b}},s,t,u) \notag \\
&=:
 \sup_{\beta_{\mathrm{b}} >0}
  \gamma_{\mathrm{b}}(\omega_{\mathrm{b}},\beta_{\mathrm{b}}). 
  \notag
\end{align}
 A point $(s,t,u)$  which achieves the minimum of the function
 $\gamma_{\mathrm{b}}(\omega_{\mathrm{b}}, \beta_{\mathrm{b}}, s, t, u)$
 is given in a solution of the following equations:
\begin{align}
&\omega_{\mathrm{b}}
 =
 \frac{s}{P_{\mathrm{b}}}\pderi{P_{\mathrm{b}}}{s}
 =
 \sum_{i\in\mathcal{L}}
 L_i \frac{(1+s)^{r-1}st^i}{1+\{(1+s)^r-1\}t^i}, \label{eq:Ps_B}\\
&\beta_{\mathrm{b}}
 =
 \frac{t}{P_{\mathrm{b}}}\pderi{P_{\mathrm{b}}}{t}
 =
 \sum_{i\in\mathcal{L}}
 \frac{L_i}{r} \frac{i\{(1+s)^r-1\}t^i}{1+\{(1+s)^r-1\}t^i},  \label{eq:Pt_B}\\
&\beta_{\mathrm{b}}
 =
  \frac{u}{Q_{\mathrm{b}}}\pderi{Q_{\mathrm{b}}}{u} 
 =
 \sum_{j\in\mathcal{R}} 
  \frac{\kappa R_j}{r} \frac{u}{f_j(u)}\frac{\partial f_j(u)}{\partial u}
 \label{eq:Qu_B}
\end{align}
The point $\beta_{\mathrm{b}}$ which gives the maximum of $\gamma_{\mathrm{b}}(\omega_{\mathrm{b}},\beta_{\mathrm{b}})$ needs to satisfy the stationary condition
\begin{align*}
 \beta_{\mathrm{b}}
  =
 (2^p-1)t u (\epsilon_{\mathrm{b}}-\beta_{\mathrm{b}}). 
\end{align*}
\end{theorem}
\begin{corollary} \label{cor:4}
 For the irregular non-binary cluster LDPC code ensemble $\Ens(N, p,r,\lambda,\rho)$ in the limit of large bit code length $n$, the following equations hold:
\begin{align*}
 &\gamma_{\mathrm{b}}(0) 
  = 
 0, \\
 &\gamma_{\mathrm{b}}(1) 
  = 
  -\epsilon_{\mathrm{b}} \log(2^p-1)
  -\kappa \frac{p}{r} 
  +\sum_{j\in\mathcal{R}}
   \frac{\kappa R_j}{r} \log \{ (2^p-1)^j + (-1)^j (2^p-1) \}.
\end{align*}
Moreover, by letting $p,r$ tend to infinity with fixed ratio,
we have
\begin{align*}
 \gamma_{\mathrm{b}}(1)
  \to
 -\frac{\kappa p}{r}.
\end{align*}
\end{corollary}
\begin{lemma}
 For $s>0$ such that Eq.~\eqref{eq:Ps_B}, \eqref{eq:Pt_B} and \eqref{eq:Qu_B} hold, we have
 \begin{align*}
  \deri{\gamma_{\mathrm{b}}}{\omega_{\mathrm{b}}}(\omega_{\mathrm{b}})
   =
  -\log s(\omega_{\mathrm{b}}).
 \end{align*}
\end{lemma}

\subsection{Analysis of Small Weight Codeword}
In this section, we investigate the growth rate of the average number of codewords of symbol and bit weight with small $\omega$.

\begin{theorem} \label{the:5}
 For the irregular non-binary cluster LDPC code ensemble $\Ens(N, p,r,\lambda,\rho)$ with $\lambda_2 >0$, 
the growth rate $\gamma(\omega)$ of the average number of codewords in terms of symbol weight, in the limit of large symbol code length for small $\omega$, is given by
\begin{align}
 \gamma (\omega)
  =
 -\frac{\omega}{r}
  \log\biggl[ \frac{2^p-1}{(2^r-1)\lambda'(0)\rho'(1)}\biggr]
 +o(\omega),
\end{align}
where we denote $f(x) = o(g(x))$ if and only if $\lim_{x \searrow 0} \bigl|\frac{f(x)}{g(x)}\bigr| = 0$
and where $\lambda'(0)\rho'(1) = \lambda_2 \sum_{j\in\mathcal{R}} (j-1)\rho_j$.

\end{theorem}
\begin{proof}
Note that for $\omega > 0$, 
\begin{align*}
 \gamma(\omega)
  =
 \gamma(0)
 +\omega \deri{^{+}\gamma}{\omega}(0) + o(\omega),
\end{align*}
where
\begin{align*}
 \deri{^{+}\gamma}{\omega}(0)
  :=
 \lim_{\omega \searrow 0} \frac{\gamma(\omega)-\gamma(0)}{\omega}
  =
 \lim_{\omega \searrow 0} \deri{\gamma}{\omega}(\omega).
\end{align*}
From Corollary \ref{cor:3}, we have $\gamma(0)=0$.
Hence, we will calculate $\lim_{\omega \searrow 0} \deri{\gamma}{\omega}(\omega)$.
From Lemma \ref{lem:2}, we have
\begin{align}
 \lim_{\omega \searrow 0} \deri{\gamma}{\omega}(\omega)
  =
 -\frac{1}{r} \lim_{\omega \searrow 0} \log \frac{s(\omega)}{2^r-1}.
 \label{eq:the5-0}
\end{align}
Recall that $s(\omega)$ satisfies Eqs.~\eqref{eq:Ps}, \eqref{eq:Pt}, \eqref{eq:Qu} and \eqref{eq:st_S}.
From Eq.~\eqref{eq:Ps}, for $\omega \searrow 0$, it holds that $s t^{i}\searrow 0$ for $i\in\mathcal{L}$.
By using this and Eq.~\eqref{eq:Pt}, we have $\beta \searrow 0$.
Notice that
\begin{align}
 f_j(u)
  =
 1+{\textstyle \binom{j}{2}}(2^p-1)u^2 + o(u^2).
 \label{eq:fj_approx}
\end{align}
By combining Eqs.~\eqref{eq:Qu} and \eqref{eq:fj_approx}, and $\beta\searrow 0$,
we get
\begin{align*}
 \beta 
  =
 \epsilon \rho'(1) (2^p -1) u^2 + o(u^2).
\end{align*}
Substituting this equation into Eq.~\eqref{eq:st_S}, we have
\begin{align}
 t
  =
 \rho'(1) u + o(u). \label{eq:th5-1}
\end{align}
The combination of this equation and $u\searrow 0$ gives $t\searrow 0$.
Since $t\searrow 0$ and $\lambda_2 >0$, from Eq.~\eqref{eq:Pt},
we get 
\begin{align*}
 \beta
  =
 \epsilon \lambda_2 s t^2 + o(t^2).
\end{align*}
Substituting this equation into Eq.~\eqref{eq:st_S}, we have
\begin{align}
 u
  =
 \frac{1}{2^p-1}\lambda_2 s t + o(t). \label{eq:th5-2}
\end{align}
Combining Eqs.~\eqref{eq:th5-1} and \eqref{eq:th5-2}, we have for $\omega \searrow 0$
\begin{align*}
 s(\omega)
  =
 (2^p-1) \frac{1}{\lambda'(0)\rho'(1)}.
\end{align*}
Thus, from Eq.~\eqref{eq:the5-0}, we obtain
\begin{align*}
 \lim_{\omega \searrow 0} \deri{\gamma}{\omega}(\omega)
   =
 \frac{1}{r}\log \biggl[ \frac{2^r-1}{2^p-1}\lambda'(0)\rho'(1) \biggr].
\end{align*}
This leads Theorem \ref{the:5}. 
\end{proof}

Similarly, the growth rate of the average number of codewords of bit weight with small weight $\omega_{\mathrm{b}}$ is given in the following theorem.
\begin{theorem} \label{the:6}
For the irregular non-binary cluster LDPC code ensemble $\Ens(N, p,r,\lambda,\rho)$ with $\lambda_2 >0$, 
the growth rate $\gamma_{\mathrm{b}}(\omega_{\mathrm{b}})$ of the average number of codewords in terms of bit weight, in the limit of large bit code length for small $\omega_{\mathrm{b}}$, is given by
\begin{align*}
 \gamma_{\mathrm{b}}(\omega_{\mathrm{b}})
  =
 -\omega_{\mathrm{b}}\log\biggl[ 
   \biggl( \frac{2^p-1}{\lambda'(0)\rho'(1)} + 1\biggr)^{1/r} - 1
  \biggr]
  +o(\omega_{\mathrm{b}}).
\end{align*}
\end{theorem}

We define
\begin{align*}
 &\delta^{*} 
  :=
 \inf\{ \omega > 0 \mid \gamma(\omega) \ge 0 \}, \\
 &\delta^{*}_{\mathrm{b}}
  :=
 \inf\{ \omega_{\mathrm{b}} > 0 \mid \gamma_{\mathrm{b}}(\omega_{\mathrm{b}}) \ge 0 \},
\end{align*}
and refer to them as the {\it normalized typical minimum distance} in terms of symbol and bit weight, respectively.
Recall that the average number of codeword of symbol weight $\omega N$ (resp.\ bit weight $\omega_{\mathrm{b}} n$) is approximated by $A(\omega N) \sim 2^{r \gamma(\omega) N}$ (resp.\ $A_{\mathrm{b}}(\omega_{\mathrm{b}} n) \sim 2^{\gamma_{\mathrm{b}}(\omega_{\mathrm{b}}) n}$).
Since $\gamma(\omega) < 0$ (resp.\ $\gamma_{\mathrm{b}}(\omega_{\mathrm{b}}) < 0$) for $\omega\in(0,\delta^*)$ (resp.\ for $\omega_{\mathrm{b}}\in(0,\delta^*_{\mathrm{b}})$), 
there are exponentially few codewords of symbol weight $\omega N$ (resp.\ bit weight $\omega_{\mathrm{b}} n$) for $\omega \in(0,\delta^*)$ (resp.\ for $\omega_{\mathrm{b}} \in (0,\delta^*_{\mathrm{b}})$).

Theorem \ref{the:5} and \ref{the:6} gives the following corollary.
\begin{corollary} \label{cor:5}
For the irregular non-binary cluster LDPC code ensemble $\Ens(N,p,r,\lambda,\rho)$ with sufficiently large $N$, 
the normalized typical minimum distances $\delta^*$ and $\delta^*_{\mathrm{b}}$ in terms of symbol and bit weight, respectively, are strictly positive if
\begin{align}
 \lambda'(0)\rho'(1) < \frac{2^p-1}{2^r-1}. \label{eq:rmd_cond}
\end{align}
\end{corollary}
\begin{remark}
For the non-binary LDPC code ensembles defined over finite field $\mathbb{F}_{2^p}$,
the normalized typical minimum distances are strictly positive if $\lambda'(0)\rho'(1)<1$
\cite{Kasai_nbw}.
For the non-binary LDPC code ensembles defined by the parity check matrices over general linear group $\mathrm{GL}(p,\mathbb{F}_2)$,
a necessary condition that the normalized typical minimum distances are strictly positive is also $\lambda'(0)\rho'(1)<1$ from Corollary \ref{cor:5} with $p=r$.
On the other hand, in the case for the non-binary cluster LDPC code ensembles,
a necessary condition that the normalized typical minimum distances are strictly positive  depends on not only $\lambda'(0)\rho'(1)$ but also the size of cluster $p,r$ as in Corollary \ref{cor:5}.

Therefore, for any degree distribution pair ($\lambda, \rho$), 
we are able to satisfy Eq.~\eqref{eq:rmd_cond} by sufficiently large $p,r$ with fixed ratio, i.e., for fixed designed rate and degree distribution pair.
\end{remark}

\subsection{Numerical Examples}

\begin{figure}[t]
 \begin{center}
  \includegraphics[width = .65\linewidth]{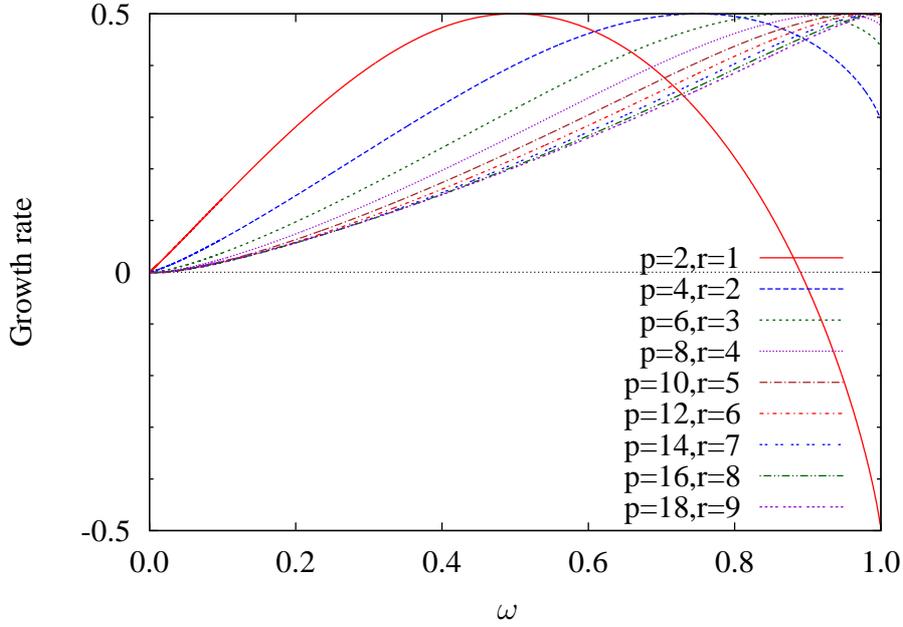} 
  \caption{Growth rates to the average symbol weight distributions for the $(2,8)$-regular non-binary cluster LDPC code ensembles with the cluster size $(p,r) = (2,1), (4,2),\dots, (18,9)$. \label{fig:symbol}}
 \end{center}
\end{figure}
\begin{figure}
 \begin{center}
  \includegraphics[width = .65\linewidth]{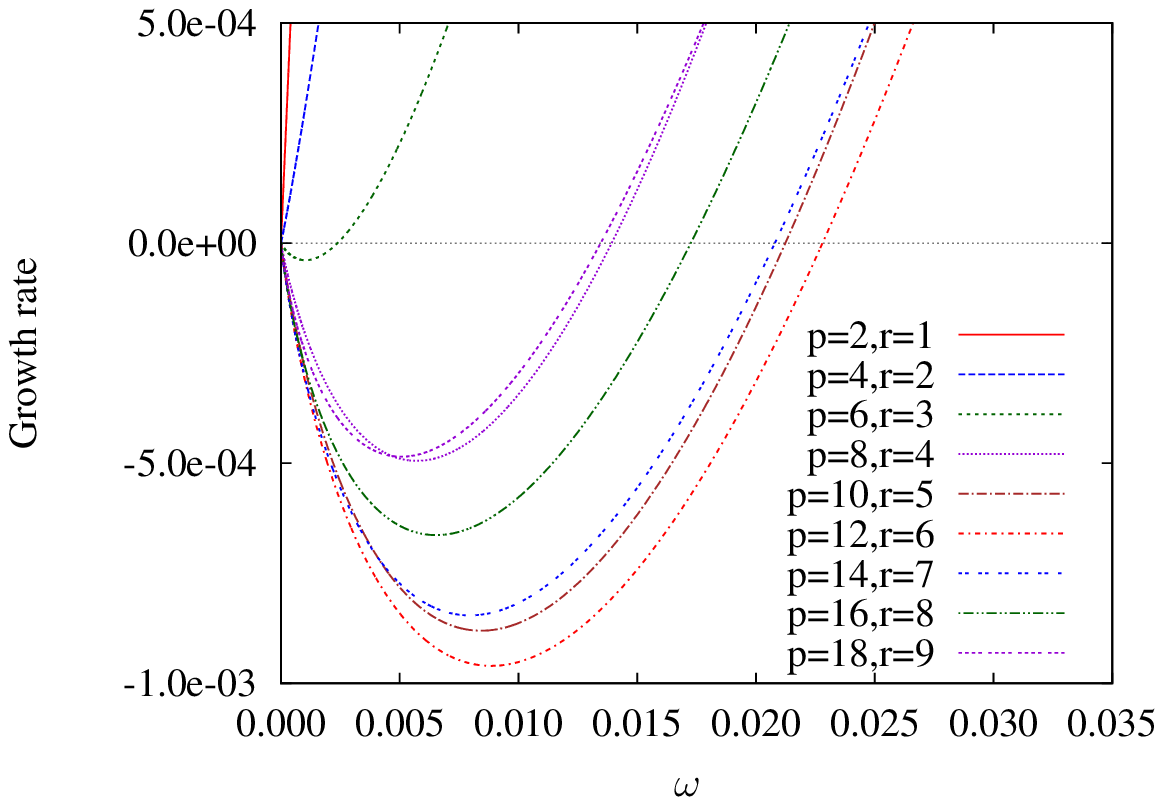}
 \caption{Growth rates to the average symbol weight distributions for the $(2,8)$-regular non-binary cluster LDPC code ensembles with the cluster size $(p,r) = (2,1), (4,2),\dots, (18,9)$.  \label{fig:symbol_d}}
 \end{center}
\end{figure}
\begin{figure}[t]
 \begin{center}
  \includegraphics[width = .65\linewidth]{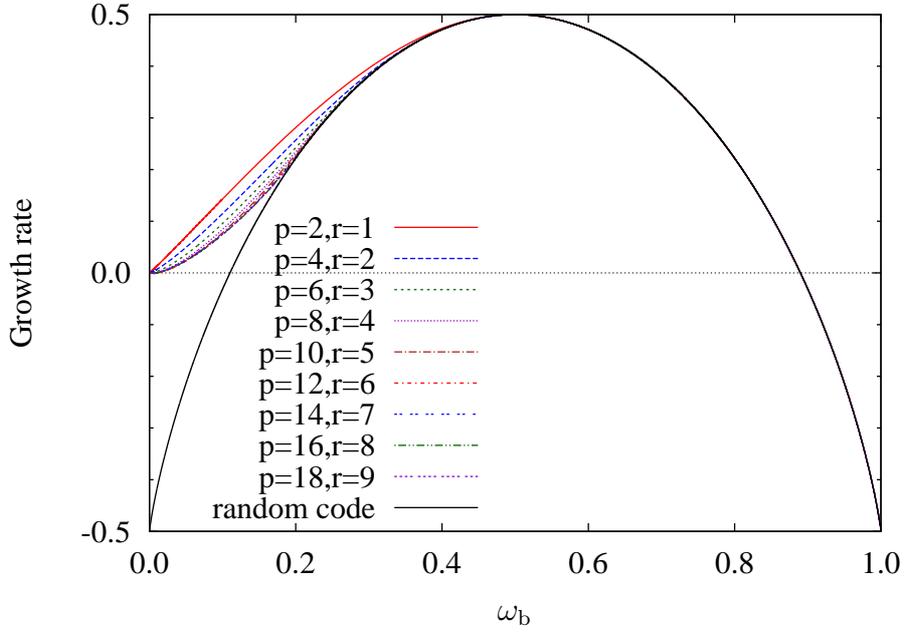}
  \caption{Growth rates to the average bit weight distributions for the $(2,8)$-regular non-binary cluster LDPC code ensembles with the cluster size $(p,r) = (2,1), (4,2),\dots, (18,9)$.
The black solid curve (random code) gives the growth rate for the binary random code ensemble of rate $0.5$.  \label{fig:bit}}
 \end{center}
\end{figure}
\begin{figure}[t]
 \begin{center}
  \includegraphics[width = .65\linewidth]{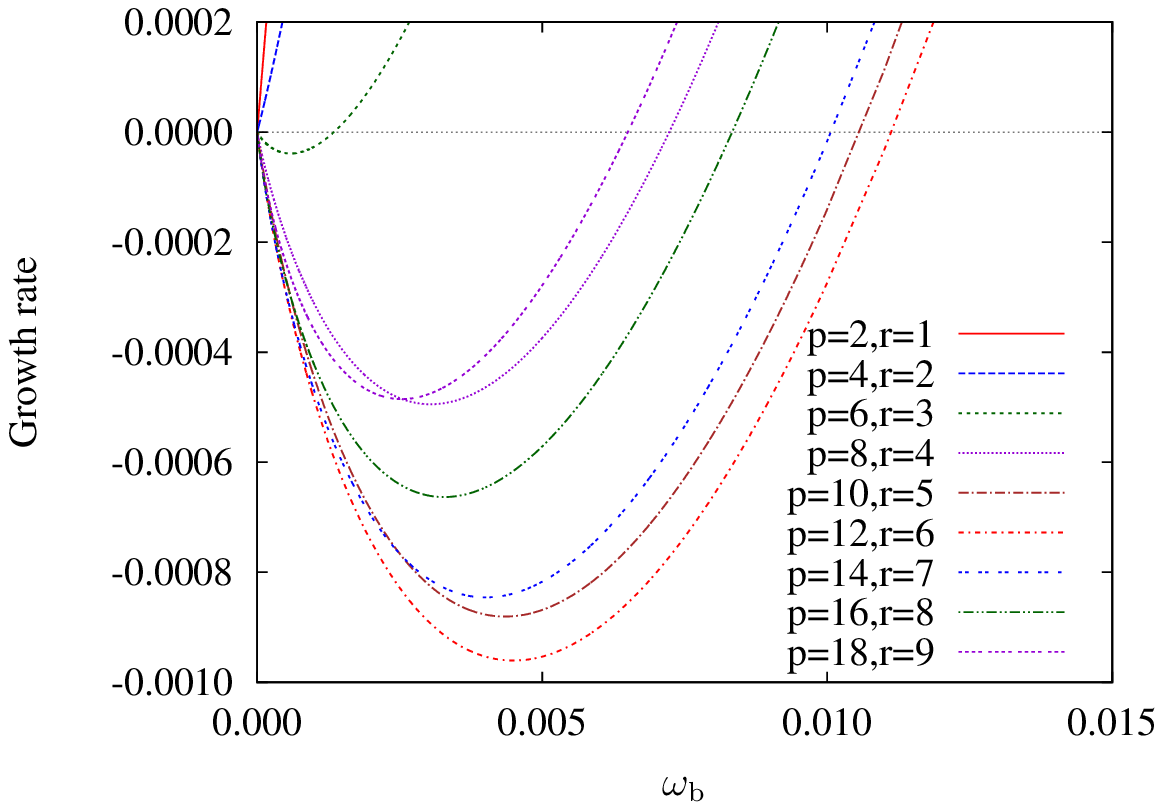}
  \caption{Growth rates to the average bit weight distributions for the $(2,8)$-regular non-binary cluster LDPC code ensembles with the cluster size $(p,r) = (2,1), (4,2),\dots, (18,9)$.  \label{fig:bit_d}}
 \end{center}
\end{figure}
\begin{figure}
 \begin{center}
  \includegraphics[width = .65\linewidth]{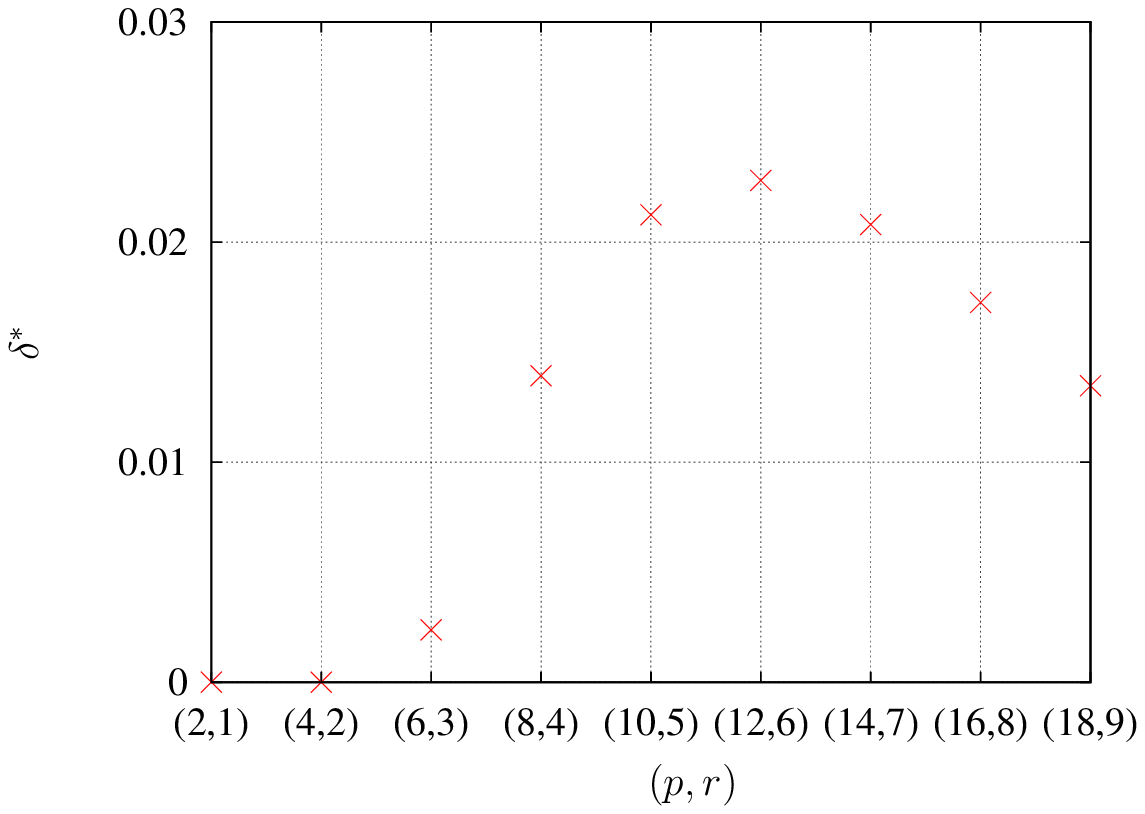}
  \caption{The normalized typical minimum distance $\delta^{*}$ of the symbol weight distribution for the (2,8)-regular non-binary cluster LDPC code ensemble with the cluster size $(p,r)=(2,1),(4,2),\dots,(18,9)$. \label{fig:zc_s}}
 \end{center}
\end{figure}
\begin{figure}
 \begin{center}
  \includegraphics[width = .65\linewidth]{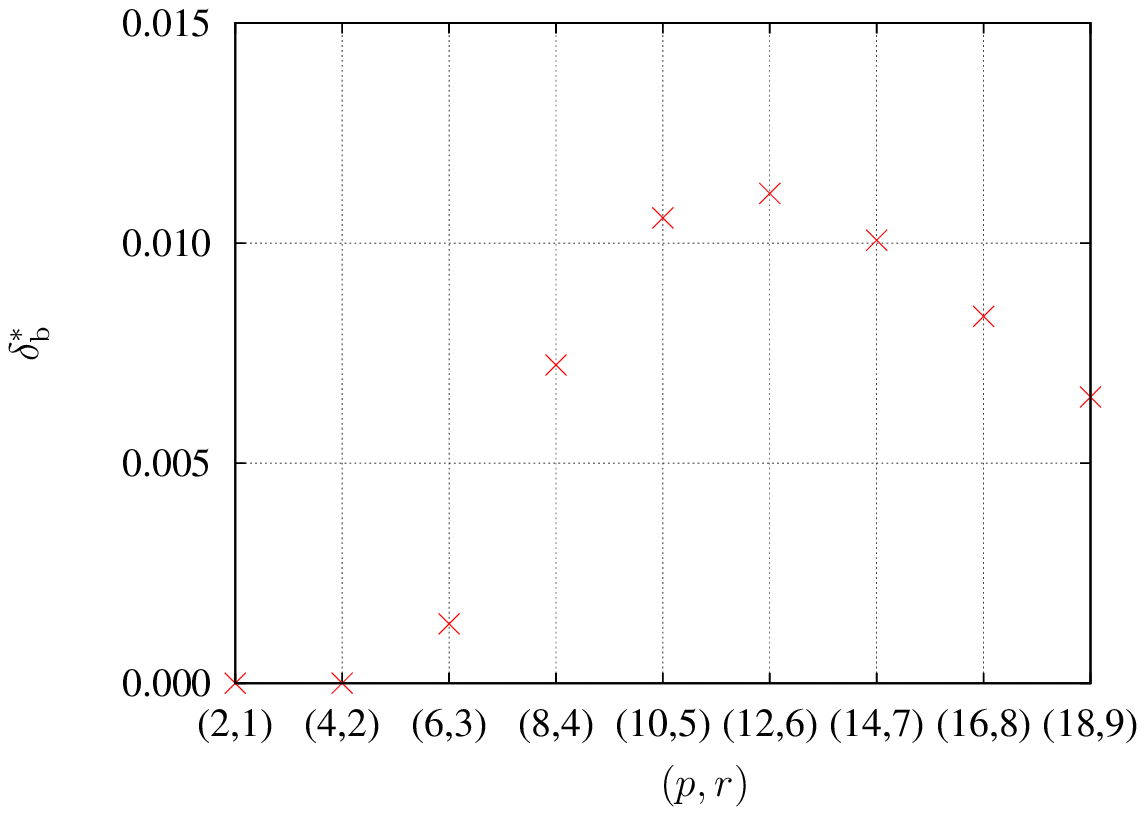}
  \caption{The normalized typical minimum distance $\delta^{*}_{\mathrm{b}}$ of the bit weight distribution for the (2,8)-regular non-binary cluster LDPC code ensemble with the cluster size $(p,r)=(2,1),(4,2),\dots,(18,9)$. \label{fig:zc_b}}
 \end{center}
\end{figure}


In this section, we show some numerical examples of the growth rates for the cluster non-binary LDPC code ensembles.
As an example, we employ the $(2,8)$-regular non-binary cluster LDPC codes.
To keep the design rate at half, we fix the ratio of the cluster size as $p/r = 2$.

Figures \ref{fig:symbol} and \ref{fig:symbol_d} give the growth rates to the average symbol weight distributions for the cluster size $(p,r) = (2,1), (4,2),\dots, (18,9)$.
As shown in Corollary \ref{cor:3}, $\gamma(1)$ tends to the design rate $0.5$.
From Figure \ref{fig:symbol_d}, we see that the slop of the growth rate at $\omega = 0$ are negative and the normalized typical minimum distance $\delta^*$ is strictly positive for $(p,r) = (6,3),(8,4),\dots,(18,9)$.
This confirms Corollary \ref{cor:5}.

Figures \ref{fig:bit} and \ref{fig:bit_d} give the growth rates to the average bit weight distributions for the cluster size $(p,r) = (2,1), (4,2),\dots, (18,9)$.
The black solid curve in Figure \ref{fig:bit} shows the growth rate of the binary random code ensemble of rate $0.5$.
As shown in Corollary \ref{cor:4}, $\gamma_{\mathrm{b}}(1)$ tends to $-0.5$.
Moreover, we see that the curves in $\omega_{\mathrm{b}} > 1/2$ converge to the growth rate of the binary random code ensemble.
From Figure \ref{fig:bit_d}, we see that the slop of the growth rate at $\omega_{\mathrm{b}} = 0$ are negative and the normalized typical minimum distance $\delta^*_{\mathrm{b}}$ is strictly positive for $(p,r) = (6,3),(8,4),\dots,(18,9)$.
This confirms Corollary \ref{cor:5}.

Figures \ref{fig:zc_s} and \ref{fig:zc_b} give the normalized typical minimum distance $\delta^{*}$ and $\delta^{*}_{\mathrm{b}}$ of the symbol and bit weight distribution, respectively, 
for the cluster size $(p,r)=(2,1),(4,2),\dots,(18,9)$.
From Figures \ref{fig:zc_s} and \ref{fig:zc_b}, we see that the normalized typical minimum distances $\delta^{*}$ and $\delta^{*}_{\mathrm{b}}$ does not monotonically increase with the size of cluster $(p,r)$.
In this case, the normalized typical minimum distances $\delta^*, \delta^{*}_{\mathrm{b}}$ have the local maximum at $(p,r) = (12,6)$.

\section{Conclusion}
In this paper, we have derived the average weight distributions for the irregular non-binary cluster LDPC code ensembles.
Moreover, we have given the exponential growth rate of the average weight distribution in the limit of large code length.
We have shown that there exist ($2,d_{\mathrm{c}}$)-regular non-binary cluster LDPC code ensembles whose normalized typical minimum distances are strictly positive.

\section*{Acknowledgment}
This work was supported by Grant-in-Aid for JSPS Fellows.
The work of K.~Kasai was supported by the grant from the Storage Research Consortium.

\bibliographystyle{IEEEtran}
\bibliography{IEEEabrv,nozaki_bib}

\end{document}